%% file: main_arxiv.tex
\documentclass[onecolumn, draftclsnofoot,11pt]{IEEEtran}

\usepackage{graphicx}
\usepackage{subfigure}
\usepackage{physics}
\usepackage{amsfonts}
\usepackage{amsthm}
\usepackage[colorlinks=true, allcolors=blue]{hyperref}
\usepackage[utf8]{inputenc} 
\usepackage[T1]{fontenc}
\usepackage{url}
\usepackage{ifthen}
\usepackage{cite}

\input{header}

\begin{document}
\title{\huge Enhancing Quantum Expectation Values via Exponential Error Suppression and CVaR Optimization} 




\author{
\IEEEauthorblockN{Touheed Anwar Atif, Reuben Blake Tate, and Stephan Eidenbenz\\}
\IEEEauthorblockA{Los Alamos National Laboratory \\
tatif@lanl.gov, rtate@lanl.gov, and  eidenben@lanl.gov 
}}

\maketitle


\begin{abstract}
Precise quantum expectation values are crucial for quantum algorithm development, but noise in real-world systems can degrade these estimations. While quantum error correction is resource-intensive, error mitigation strategies offer a practical alternative. This paper presents a framework that combines Virtual Channel Purification (VCP) technique with Conditional Value-at-Risk (CVaR) optimization to improve expectation value estimations in noisy quantum circuits.
Our contributions are twofold: first, we derive conditions to compare CVaR values from different probability distributions, offering insights into the reliability of quantum estimations under noise. Second, we apply this framework to VCP, providing analytical bounds that establish its effectiveness in improving expectation values, both when the overhead VCP circuit is ideal (error-free) and when it adds additional noise. By introducing CVaR into the analysis of VCP, we offer a general noise-characterization method that guarantees improved expectation values for any quantum observable. We demonstrate the practical utility of our approach with numerical examples, highlighting how our bounds guide VCP implementation in noisy quantum systems.
\end{abstract}

\section{Introduction}

 


The accurate computation of quantum expectation values plays a crucial role in the development and verification of a plethora of quantum algorithms \cite{peruzzo2014variational,farhi2014quantum,wecker2015progress,mcclean2016theory,biamonte2017quantum}. However, real-world quantum systems are often plagued by various types of noise, including gate errors, measurement inaccuracies, and decoherence, which can significantly affect the accuracy of these estimations. Although quantum error correction has been proposed as a means to combat this noise by encoding logical qubits in multiple physical qubits \cite{lidar2013quantum}, the overhead required for error correction is too large for near-term devices \cite{lee2021even}. A near term alternative is to employ error mitigation strategies, which improve the accuracy of quantum computations by obtaining valuable information from noisy quantum circuits without the need to physically recover a noiseless state \cite{cai2023quantum,temme2017error,kim2023evidence}.



In this paper, we introduce a framework that enhances quantum expectation value estimations by combining exponential error suppression techniques \cite{huggins2021virtual,koczor2021exponential,liu2024virtual} with Conditional Value-at-Risk ($\tcvar$) optimization \cite{acerbi2002coherence}. Specifically, we focus on a recent error suppression method called Virtual Channel Purification (VCP), introduced by Liu et al \cite{liu2024virtual}, that aims to suppress noise in quantum circuits without requiring detailed knowledge of error models, input/output states, or specific problem contexts.
The authors here primarily demonstrated the technique's effectiveness by looking at the fidelity of quantum states. However, VCP as a virtual technique is most impactful when applied to the expectation values of observables, a perspective that we explore in this work.

Conditional Value-at-Risk ($\tcvar$), also known as Expected Shortfall, is a risk measure that considers the tail risks of a probability distribution. In the respect of quantum algorithms, the authors in \cite{barkoutsos2020improving} 
proposed a modification to the classical optimization algorithm using $\tcvar$ and showed performance improvements for VQE and QAOA. In \cite{barron2024provable}, the authors demonstrated the robustness of $\tcvar$ against noise toward generating better expectation values and provided bounds on noise-free expectation values. 



Our contributions are twofold. First, we derive conditions on two probability distributions that enable a comparison of the $\tcvar$ values produced by samples from each distribution. This comparison is crucial for assessing the reliability of quantum expectation value estimations in the presence of noise, providing insights into which distributions leads to more robust outcomes. Second, we apply this method to the study of the VCP technique. In the idealized regime where the purification circuit, particularly the additional swap network, is error-free, and only the unitary being purified has incoherent noise, we provide analytical bounds that establish the efficacy of the VCP method in enhancing the precision of expectation values. The assumption on incoherent noise is valid in most practical scenarios, particularly with the application of techniques like Pauli-twirling \cite{koczor2021dominant,dalzell2024random,liu2024virtual}.
When noise is present in the swap network, which introduces additional noise the purification circuit, we extend our analysis to provide inner bounds that characterize the region where the technique offers an advantage in terms of error parameters. 

By introducing CVaR into the analysis of VCP, as  our first result for noise-free swap network we provide a provable guarantee for improving expectation values of any quantum observable  diagonal in the computational basis (Theorem \ref{thm:noisy_vcp}). 
In Theorem \ref{thm:vcp_comp}, we analyze and compare different orders of VCP circuits in terms of their ability to reduce noise using $\tcvar$. We extend this analysis to a multilayer generalization in Theorem \ref{thm:vcp_comp_genralized}.
Next, we consider the case where the swap network of the VCP circuit introduces additional noise. Here, we first provide an exact form for the effective virtual channel when a Clifford is purified employing the VCP circuit that uses a noisy swap circuit (Theorem \ref{thm:noisy_vcp_ratio_obs}). Its generalization to higher orders is described in Theorem \ref{thm:noisy_vcp_ratio_obs_orderM}. Building on these theorems, our final result in Theorem \ref{thm:vcp_noisy} characterizes an inner bound for the parameter region where the VCP technique remains advantageous despite a noisy swap network.

We further demonstrate the practical implications of our findings through numerical examples. The first example (Section \ref{sec:example1}) examines the region defined by Theorem \ref{thm:vcp_noisy} for depolarizing noise models and provides a characterization in terms of the depolarizing parameters $p$ and $q$ corresponding to the swap network noise and Clifford noise, respectively. The second example (Section \ref{sec:example2}) explores a scenario involving a collection of IID Clifford gates, highlighting the advantageous region in terms of the number of gates and the associated error parameters.
These examples
showcase how our bounds can guide the implementation of VCP in realistic quantum systems and improve its performance under noise. 

The paper is organized as follows: Section \ref{sec:prelims} presents the preliminaries on exponential error suppression and sets up the notation.
Section \ref{sec:main_results} provides the main results focusing on the application of our results to the VCP technique, deriving bounds and characterizing the error parameters. Section \ref{sec:examples} provides numerical examples to illustrate our findings and discusses their practical implications. Finally, Section \ref{sec:conclusion} concludes with a discussion on future directions for research and potential applications of the proposed methodology in quantum computing. The proofs are provided in the appendix.

\section{Preliminaries and Notation}\label{sec:prelims}

\subsection{Virtual Channel Purification (VCP)}

VCP builds on the concept of exponential error mitigation, specifically the virtual state preparation technique. It uses multiple copies of a noisy quantum channel, rather than a state, to implement a purified channel, with fidelity improving exponentially with the number of copies. The technique combines ideas from the Choi–Jamiołkowski isomorphism and flag fault-tolerance to enable noise reduction without physically purifying the state or channel, relying instead on classical postprocessing and standard quantum circuits. This makes VCP an efficient, hardware-friendly solution for quantum error mitigation, particularly in tasks where accurate expectation values are crucial.
Formally, we have the following.
Consider a case where the ideal operation is $\calU$, but in practice, the implemented operation is the noisy version $\calU_\calE = \calE \calU$, where 
\[
\calE = p_0 \calI + \sum_{i=1}^{4^N-1}p_i \calE_i,
\]
represents the noise that distorts the ideal operation, $N$ is the number of qubits on which $\calU$ acts, $\calI$ denotes the identity channel, and $\calE_i(\rho) = E_i \rho E_i^\dagger$ are Pauli noise components. It is assumed that $p_0 > p_i$ for all $i$. The objective of VCP is to implement $\calU_{\calE^{(M)}} = \calE^{(M)}\calU$, where $\calE^{(M)}$ is a purified noise channel defined as
\[
\calE^{(M)} = \frac{1}{\sum_{i=0}^{4^N-1}p_i^M} \left(p_0^M \calI + \sum_{i=0}^{4^N-1} p_i^M \calE_i\right).
\]


\subsection{Conditional Value-at-Risk}
\begin{mydefinition}[Conditional Value-at-Risk]
    Consider an integrable real-valued random variable $X$ with CDF $F_X : R \rightarrow [0,1]$. Then, the (lower) CVaR at level $\alpha \in (0,1]$ is defined as 
    \[
    \cvar_\alpha[X] := \alpha^{-1} \EE[X;X\leq x_\alpha] + x_{\alpha}(1-\PP(X \leq x_\alpha)),
    \]
    where $x_{\alpha} := inf_{x \in \RR}[F_X(x) \geq \alpha]. $
\end{mydefinition}


In the case when $F_X(x_\alpha) = \alpha$, $\cvar$ simplifies to $\cvar_\alpha[X] = \EE[X|X\leq x_\alpha]$, i.e., the expectation of $X$ when  $X$ is conditioned to take values in its bottom $\alpha$ quantile. The upper $\cvar$ is defined as
\[
\bcvar_\alpha[X] = -\cvar_\alpha[-X].
\]

\textbf{Notations.} The set of density operators on Hilbert space $\calH_A$ is denoted by $\calD(A)$. Let $[\Theta] := \{1,2,\cdots,\Theta\}$.

\section{Main Results} \label{sec:main_results}

\subsection{Conditional Value-at-Risk}

\begin{mylemma}\label{lem:cvarIneq}
    Consider three random variables $X, X_1,$ and $X_2$, defined on a set $\calX$ with distributions $P_X, P_{X_1},$ and $P_{X_2}$, respectively. For positive constants $C_1 \geq  C_2 \geq 1$, the given distributions are such that for all $x\in \calX$,
    \begin{align}
        \frac{P_X(x)}{C_1} &\leq P_{X_1}(x) , \quad \text{and}\quad
        \frac{P_X(x)}{C_2} \leq P_{X_2}(x) \leq \frac{C_1}{C_2} P_{X_1}(x).
    \end{align}
     Then, {with $\alpha_1 := 1/C_1$ and $\alpha_2 := 1/C_2$,} we have the following set of inequalities
     \begin{align}
         \cvar_{\alpha_1}[X_1] &\leq \cvar_{\alpha_2}[X_2] \leq \EE[X] \leq \bcvar_{\alpha_2}[X_2] \leq \bcvar_{\alpha_1}[X_1] 
     \end{align}
Therefore, the $\tcvar$ of $X_1$ establishes lower and upper bounds for the $\tcvar$ of $X_2$, which, in turn, bounds the expectation of $X$ from below and above, respectively.
\end{mylemma}
\begin{proof}
    A proof to the lemma is provided in Appendix \ref{proof:lem:cvarIneq}.
\end{proof}

\subsection{Application to Virtual Channel Purification}
Let $\rho_0 \in \calD(A)$ and $\rho$ is defined as $\rho := \calU(\rho_0)$. Consider a observable $O$, diagonal in the computational basis, corresponding to which the measurements are made. Let the noise $\calE$ associated with the implementation of $\calU$ be such that its identify component is dominant, i.e., $p_0 \geq p_i$ for all $i\in [4^N-1]$.
\begin{mytheorem}\label{thm:noisy_vcp}
    Let $\tilde{\rho} := \calE(\calU(\rho_0))$ denote the noisy sample of $\rho$, and $\tilde{\rho}^{(L)} 
    := \calU_{\calE^{(L)}}(\rho)$ denote the version of ${\rho}$ that has been purified using VCP of order-$L$ ($L > 1$). Let $X, {X}_{noisy}$,  and ${X}_{vcp}^{L}$ denote the  random variables corresponding to the measurement of $\rho , \tilde\rho,$ and $\tilde\rho^{(L)}$, with respect to $O$, respectively. We then have:
    \begin{align}
         \cvar_{p_0}[{X}_{noisy}] &\leq \cvar_{\alpha_L}[{X}_{vcp}^{L}] 
         \leq \EE[X] \leq \bcvar_{\alpha_{L}}[{X}_{vcp}^{L}]  \bcvar_{p_0}[{X}_{noisy}],
     \end{align}    
     where $\alpha_L = \frac{p_0^L}{\sum_{i=0}^{4^N-1}p_i^L}$.
\end{mytheorem}
\begin{proof}
    A proof is provided in Appendix \ref{proof:thm:noisy_vcp}.
\end{proof}
\begin{mytheorem}\label{thm:vcp_comp}
    Define the noisy versions of $\rho$ that are purified using an order-$L$ and an order-$M$ VCP as  $\tilde{\rho}^{(L)} = \calU_{\calE^{(L)}}(\rho)$ and $\tilde{\rho}^{(M)} = \calU_{\calE^{(M)}}(\rho)$, respectively, for $L < M$. Let $X, \tilde{X}_L$,  and $\tilde{X}_M$ denote the  random variables corresponding to the measurement of $\rho , \tilde\rho^{(L)},$ and $\tilde\rho^{(M)}$, with respect to $O$. We then have the following relation:
    \begin{align}
         \cvar_{\alpha_L}[\tildeX_L] &\leq \cvar_{\alpha_M}[\tildeX_{M}] \leq \EE[X] \leq \bcvar_{\alpha_{M}}[\tildeX_{M}] \leq \bcvar_{\alpha_L}[\tildeX_L],
     \end{align}    
     for all $1\leq L < M$, where $\alpha_{x} = \frac{p_0^{x}}{\sum_{i=0}^{4^N-1}p_i^{x}}$.
\end{mytheorem}
\begin{proof} 
A proof is provided in Appendix \ref{proof:thm:vcp_comp}.
\end{proof}


As a next result, we consider a multi-layer generalization of the VCP where the circuit consists of $k$ noisy Unitaries $\calU_{\calE_i} := \calE_i \calU_i$ that are applied on the input state $\rho_0$, where $\calU_i$ is a noiseless Unitary, and \[\calE_i := p_{i_0}\calI + \sum_{j=1}^{4^N-1}p_{i_j}\calE_{i_j}.\] The Unitary $\calU_{\calE_i}$ is purified using an order-$l_i$ VCP that implements $\calU_{\calE^{(l_i)}_i} := \calE^{(l_i)}_i  \calU_i$, and hence the vector $\vec{l} := \{l_1,l_2,.. l_k\}$ characterizes the VCP purification order for each of the noisy Unitaries.

\begin{mytheorem}\label{thm:vcp_comp_genralized}
    For $\rho_0 \in \calD(A)$, let $\rho = \calU_1\circ\calU_{2}\cdots \circ \calU_k(\rho_0)$. Define the noisy version of $\rho$ that is purified using an order-$\vec{l}$ VCP as  $\tilde{\rho}^{(\vec{l})} := \calU_{\calE^{(l_1)}}\circ\calU_{\calE^{(l_2)}}\cdots \circ \calU_{\calE^{(l_k)}}(\rho)$. Similarly, define $\tilde{\rho}^{(\vec{m})} := \calU_{\calE^{(m_1)}}\circ\calU_{\calE^{m_2)}}\cdots \circ \calU_{\calE^{(m_k)}}(\rho)$ 
    as an order-$\vec{m}$ VCP of $\rho$, for $\vec{m}\neq \vec{l}$ .
    Further, let $X, \tilde{X}_{\vec{l}}$,  and $\tilde{X}_{\vec{m}}$ denote the  random variables
    corresponding to the measurement of $\rho , \tilde\rho_{\vec{l}},$ and $\tilde\rho_{\vec{m}}$, with respect to $O$. If $\vec{l}$ and $\vec{m}$ are such that 
    \begin{align}
      \Pi_{i=1}^k \round{\frac{p_{i_{j_i}}}{p_{i_0}}}^{m_i-l_i} \leq 1, \label{eq:m_l_condition}
    \end{align}
    for all $j_i\in [1,4^N-1]$ and  $i\in [1,k]$, then
    we have the following relation:
    \begin{align}
         \cvar_{\alpha_{\vec{l}}}[\tildeX_{\vec{l}}] &\leq \cvar_{\alpha_{\vec{m}}}[\tildeX_{\vec{m}}] 
         \leq \EE[X] \leq \bcvar_{\alpha_{\vec{m}}}[\tildeX_{\vec{m}}]\leq 
          \bcvar_{\alpha_{\vec{l}}}[\tildeX_{\vec{l}}], \label{eq:m_l_cvar}
     \end{align}    
 where $$\alpha_{\vec{l}} := \frac{\prod_{i=1}^k p_{i_0}^{l_i} }{\prod_{i=1}^k \round{\sum_{j=0}^{4^N-1}p_{i_j}^{l_i}}}.$$
\end{mytheorem}
\begin{proof} 
The proof follows from similar arguments as the proof of Theorem \ref{thm:vcp_comp}, so is omitted for brevity.
\end{proof}
\begin{myremark}
    If $\vec{l}$ and $\vec{m}$ are such that $m_i \geq l_i$ for all $i \in [1,k]$, then, we can assert that the inequality \eqref{eq:m_l_condition} holds for any $\calE$ whose leading component is identity. As a result, we have  inequality \eqref{eq:m_l_cvar} satisfied. The above theorem also yeilds Theorem \ref{thm:vcp_comp} as a special case. It is stated later for pedagogical ease.
\end{myremark}

\begin{figure}
    \centering
    \includegraphics[width=\linewidth]{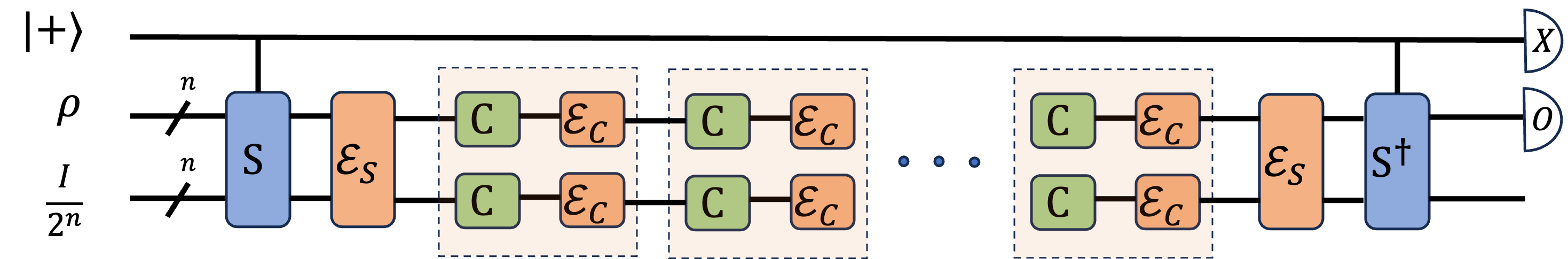}
    \caption{VCP circuit of order-$2$ in presence of Noisy Swap gates with their noise modeled by $\calE_S$. The objective is to purify the Clifford gates with each having noise $\calE_C$.}
    \label{fig:noisy_swap}
\end{figure}

\subsection{Application to Noisy Virtual Channel Purification}
 We now provide results pertaining to circuits that contain noisy swap gates, as shown in Figure \ref{fig:noisy_swap}. 
 To make the analysis tractable, we assume that each noisy Unitary $U_\calE$ that we wish to purify is a Clifford Unitary $C$ followed by a Pauli channel $\calE$; $U_\calE = C \circ \calE$. With the Clifford assumption, we can move the clifford gate before the CSwap noise, without affecting the nature of the noise, i.e., the noise remains Pauli-noise, as shown in Figure \ref{fig:grouped_clifford}. 
 More precisely, for an arbitrary $\rho \in \calD(A)$, and a Pauli-noise channel $\calP$
\begin{align}
    C\calP(\sigma)C^\dagger = \calP'(C\sigma C^\dagger), 
\end{align}
where 
\begin{align}
    \calP' = p_0 \calI +  \sum_{i=1}^{4^N-1} p_i \calE_{\pi(i)}= p_0 \calI + \sum_{i=1}^{4^N-1} p_{\pi^{-1}(i)} \calE_i,
\end{align}
for some a permutation operator $\pi$ on $P_n\backslash \{I^{\otimes n}\}$.

Further, we also note that the composition of two Pauli-noise channels gives another Pauli-noise channel. 
\begin{figure}[!htb]
    \centering
    \includegraphics[width=0.8\linewidth]{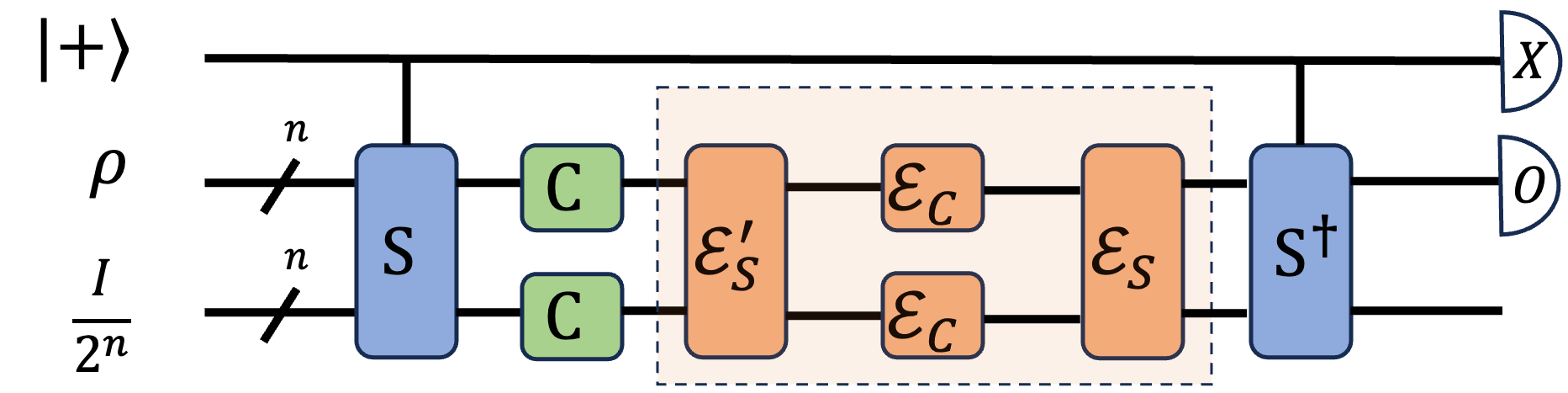}
    \caption{The result of conjugating the Pauli-noise of the CSwap gate with the Clifford C.}
    \label{fig:grouped_clifford}
\end{figure}
Now we compute the expectations required in the VCP protocol when the Control-Swap (CSwap) gates have pauli noise, and provide a condition that ensures the VCP can overcome the CSwap noise and still purify the Unitary (Clifford) gate. 

\begin{mytheorem}\label{thm:noisy_vcp_ratio_obs}
    For the given input state $\rho$ on $N$ qubits, and for the schematic in Figure \ref{fig:grouped_clifford} that implements an order-$2$ VCP of the Clifford $C$ while accounting for the CSwap and Circuit noise, we obtain the following:
    \begin{align}
        \frac{\langle X\otimes O\rangle}{\langle X\otimes I_{2^N}\rangle} = \Tr{O \calE_{\mu}(\rho)}, 
    \end{align}
    where
    \begin{align}
        \calE_{\mu} &:= \mu_0\calI + \sum_{i=1}^{4^N-1}\mu_i\calE_i,\nonumber\\
        \mu_k &:= \frac{\sum_{a,b,c,d \in \{0,1\}^{4^{N}}} p_{ab}p_{cd}q_{a\oplus c \oplus k}q_{b\oplus d \oplus k}}{\sum_{a,b,c,d,k' \in \{0,1\}^{4^{N}} } p_{ab}p_{cd}q_{a\oplus c \oplus k'}q_{b\oplus d \oplus k'}}.\label{eq:mu_def}
    \end{align}
    and $\{p_{i,j}\}$ and $\{q_k\}$ denote the probabilities in a Pauli channel decomposition for the the CSwap network implemented on $2N$ qubits and the circuit noise on $N$ qubits, respectively.
\end{mytheorem}
\begin{proof}
    A proof is provided in Appendix \ref{proof:thm:noisy_vcp_ratio_obs}.
\end{proof}

A generalization of the above theorem to an order-$M$ VCP yields the following. 
\begin{mytheorem}\label{thm:noisy_vcp_ratio_obs_orderM}
    For the given input state $\rho$ on $N$ qubits, and for a circuit that implements an order-$M$ VCP of the Clifford $C$ and accounts for the CSwap and Circuit noise, we obtain the following:
    \begin{align}
        \frac{\langle X\otimes O\rangle}{\langle X\otimes I_{2^N}\rangle} = \Tr{O \calE^{(M)}_{\mu}(\rho)}, 
    \end{align}
    where
    \begin{align}
        \calE^{(M)}_{\mu} &:= \mu_0^{(M)}\calI + \sum_{i=1}^{4^N-1}\mu_i^{(M)}\calE_i \nonumber \\
    \mu_k^{(M)} &:= \frac{\tau_k^{(M)}}{\sum_{k \in\{0,1\}^{4^{N}}\tau_k^{(M)}}} \nonumber \\
    \tau_k^{(M)} &:= \hspace{-0.2in}\sum_{\substack{a_1,\cdots a_M \in \{0,1\}^{4^{N}}  \\ b_1,\cdots b_M \in \{0,1\}^{4^{N}} } } \hspace{-0.2in}p_{a_1,\cdots a_M}p_{b_1,\cdots b_M} \prod_{i=1}^M q_{a_i\oplus b_i \oplus k},
    \end{align}
    and $\{p_{a_1,\cdots,a_M}\}$ and $\{q_k\}$ denote the probabilities in a Pauli channel decomposition for the CSwap network over $MN$ qubits and the circuit noise on $N$ qubits, respectively.
\end{mytheorem}
\begin{proof}
   The proof follows from similar arguments as the proof of Theorem \ref{thm:noisy_vcp_ratio_obs} so is omitted for the sake of brevity.
\end{proof}

Application of the above theorems to $\cvar$ gives the following theorem.

\begin{mytheorem}\label{thm:vcp_noisy}
Let $\bar{X}_{vcp}$ denote the random variable corresponding to the measurement of the output of a VCP circuit that implements a noisy CSwap Unitary, with noise modeled by $\calE_S$. 
    For the given $\calE_S$ and $\calE_C$, if the resulting $\calE_{\mu}$ is such that $\mu_0 \geq q_0$ and $\frac{\mu_i}{\mu_0} \leq \frac{q_i}{q_0}$ for all $i \in [1,4^N-1]$, we have 
     \begin{align}
         \cvar_{q_0}[{X}_{noisy}] &\leq \cvar_{\mu_0}[\bar{X}_{vcp}] 
         \leq \EE[X] \leq \bcvar_{\mu_0}[\bar{X}_{vcp}] \leq 
          \bcvar_{q_0}[{X}_{noisy}] ,
     \end{align}    
     where $X$ and ${X}_{noisy}$ are as defined earlier.
\end{mytheorem}
\begin{proof}
    A proof is provided in Appendix \ref{proof:thm:vcp_noisy}.
\end{proof}
The generalization of the above theorem to arbitrary orders follows directly and is straightforward to establish.


\section{Examples}\label{sec:examples}
In this section, we discussion the practical implications of our findings with numerical examples.
\input{example_arxiv}


\section{Conclusion}\label{sec:conclusion}
In this work, we introduced a framework that combines VCP with CVaR optimization to enhance the quantum expectation value estimations in the presence of noise. By deriving conditions for comparing CVaR values across distributions and applying this to VCP, we established analytical bounds that demonstrate the method’s effectiveness in improving estimation accuracy.
Our approach, which relies solely on the noise profile of the quantum circuit, offers a general, robust error mitigation strategy applicable to a wide range of quantum observables. Numerical examples further illustrate the practical benefits of our method.
As future directions, we aim to extend these results to deeper circuits and derive bounds that can pave the way for more reliable quantum computations in near-term quantum systems. Additionally, we plan to experiment with the ideas presented on real quantum computers and explore the challenges in validating the theoretical findings.

\section{Acknowledgment}
This material was supported by the U.S. Department of Energy, Office of Science, Office of Advanced Scientific Computing Research, under the Accelerated Research in Quantum Computing (ARQC) program and by the U.S. Department of Energy, Office of Science, National Quantum Information Science Research Centers, Quantum Science Center
TAA was funded by the QSC to perform the analytical calculations and to write the manuscript along with the other authors, and by ARQC to investigate different examples. RBT and SE acknowledge the support of the NNSA's Advanced Simulation and Computing Beyond Moore’s Law Program at Los Alamos National Laboratory. LA-UR-25-20646.



\appendix
\section{Proof of the Lemmas and Theorems}

\subsection{Proof of Lemma \ref{lem:cvarIneq}}\label{proof:lem:cvarIneq}
    We provide the proof for the case when $F_X(x_\alpha) = \alpha$, i.e, $\cvar_{\alpha}[X] = \alpha^{-1} \EE[X;X\leq x_\alpha]$. The general case follows from extending the same arguments. The inequality relating $\tcvar$ and the expectation follows from \cite[Lemma 1]{barron2024provable}, and hence we only need to establish the relation between $\cvar_{\alpha_1}[X_1]$ and $\cvar_{\alpha_2}[X_2]$.
    The proof of the first inequality (from left) can be formulated as an optimization problem, where the PMFs $P_X$ and $P_{X_1}$ are given satisfying the constraint in the lemma, and we aim to find the optimal PMF $P_{X_2}$ that minimizes the $\cvar_{\alpha_2}[X_2]$. Observe that the $\cvar_{\alpha_2}[X_2]$ is minimized for the choice of $P_{X_2}$ that has the highest probabilities to the smallest values of $x \in \calX.$ However, a feasible $P_{X_2}$ needs to satisfy the lower and upper bounds. Therefore, the optimum can be constructed by assigning the probabilities in a fashion such that the lower bound is satisfied for all $x\in \calX$, and the remaining probability is assigned in a way that the smallest values of $x$ saturate the upper bound, until the total probability becomes $1$. Formally, we let
    \begin{align}
        P_{X_2}(x) &= 
        \begin{cases}
            \frac{C_1}{C_2}P_{X_1}(x) & \mbox{ for } x \leq x_0 \nonumber \\
            \frac{P_{X}(x)}{C_2}& \mbox{ for } x > x_0, \nonumber \\
        \end{cases}
    \end{align}
    where $x_0$ is such that
    \begin{align}
        \int_{-\infty}^{\infty} P_{X_2}(x) = \int_{-\infty}^{x_0} \frac{C_1}{C_2}P_{X_1}(x) + \int_{x_0}^{\infty} \frac{P_{X}(x)}{C_2} = 1. \nonumber
    \end{align}

    Define $x^{*}_{1}$, and $x^{*}_{2}$, as 
    \[x^{*}_{1} := \inf_{x} \{F_{X_1}(x) \geq \alpha_1\}, \quad \mbox{ and } \quad x^{*}_{2} := \inf_{x} \{F_{X_2}(x) \geq \alpha_2\} \]

    Firstly, consider the case when $x_{1}^* \leq x_0$. Then 
    \begin{align}
        \int_{-\infty}^{x_1^*} P_{X_1}(x) = \frac{1}{C_1} \implies  \int_{-\infty}^{x_1^*} \frac{C_1}{C_2}P_{X_1}(x) = \frac{1}{C_2}. 
    \end{align}
    Since $x_{1}^* \leq x_0$, the right hand side in the above inequality is in fact the PMF of $X_2$, and therefore, $x_{2}^* = x_1^*$, and 
    \begin{align}
        \cvar_{\alpha_2}[X_2] = C_2 \int_{-\infty}^{x_2^{*}} x P_{X_2}(x) = C_2 \int_{-\infty}^{x_1^{*}} x \frac{C_1}{C_2}P_{X_1}(x) =  C_1 \int_{-\infty}^{x_1^{*}} x P_{X_1}(x) =  \cvar_{\alpha_1}[X_1] 
    \end{align}

Now, we  consider the case when $x_{1}^* \leq x_0$. We begin by making the following observation: 
$$x_{1}^* \leq x_2^*.$$ 
This can be shown as follows. We know from the definitions that 

\begin{align}
    \int_{-\infty}^{x_1^*} P_{X_1}(x) = \frac{1}{C_1} &\implies \int_{x_0}^{x_1^*} P_{X_1}(x) = \frac{1}{C_1} - \int_{-\infty}^{x_0} P_{X_1}(x) \quad \mbox{and} \label{eq:defx1}\\
    \int_{-\infty}^{x_2^*} P_{X_2}(x) = \frac{1}{C_2} &\implies \int_{x_0}^{x_2^*} \frac{P_{X}(x)}{C_2} = \frac{1}{C_2} - \int_{-\infty}^{x_0} \frac{C_1}{C_2}P_{X_1}(x) \nonumber \\
    & \implies \int_{x_0}^{x_2^*} \frac{P_{X}(x)}{C_1} = \frac{1}{C_1} - \int_{-\infty}^{x_0} P_{X_1}(x).\label{eq:defx2}
\end{align}
From noting the right hand sides are the same for both the above expressions and the constraint $P_{X_1}(x) \geq \frac{P_X(x)}{C_1}$,  we have $x_{1}^* \leq x_2^*.$

Now, we have
\begin{align}
    \cvar_{\alpha_2}&[X_2] = C_2 \int_{-\infty}^{x_2^*} x P_{X_2}(x)  = C_1 \int_{-\infty}^{x_0} x P_{X_1}(x) +  \int_{x_0}^{x_1^*} x P_{X}(x)+  \int_{x_1^*}^{x_2^*} x P_{X}(x) \nonumber \\
    & = C_1 \int_{-\infty}^{x_0} x P_{X_1}(x) -  \int_{x_0}^{x_1^*} x \left({C_1}{P_{X_1}(x)} - P_{X}(x)\right)+  {C_1}\int_{x_0}^{x_1^*} x{P_{X_1}(x)} + \int_{x_1^*}^{x_2^*} x P_{X}(x) \nonumber \\
     & \overset{a}{\geq} C_1 \int_{-\infty}^{x_0} x P_{X_1}(x) -  x_1^*\int_{x_0}^{x_1^*} \left({C_1}{P_{X_1}(x)} - P_{X}(x)\right)+  {C_1}\int_{x_0}^{x_1^*} x{P_{X_1}(x)} + x_1^*\int_{x_1^*}^{x_2^*} P_{X}(x) \nonumber \\
     & = C_1 \int_{-\infty}^{x_0} x P_{X_1}(x) +  {C_1}\int_{x_0}^{x_1^*} x{P_{X_1}(x)} -  x_1^*{C_1} \left[\int_{x_0}^{x_1^*} {P_{X_1}(x)}  - \int_{x_0}^{x_2^*} \frac{P_{X}(x)}{C_1}\right] \nonumber \\
      & \overset{b}{=} C_1 \int_{-\infty}^{x_1^*} x P_{X_1}(x) = \cvar_{\alpha_1}[X_1],
\end{align}

where $(a)$ follows from $P_{X_1}(x) \geq \frac{P_X(x)}{C_1}$ and using the argument that for any $f:\calX \rightarrow \RR$, and $a<b \in \RR$, if $f(x) \geq 0 \; \forall \; x \in [a,b] $, then $$a\int_a^b f(x) dx  \leq \int_a^b x f(x) dx \leq b\int_a^b f(x) dx,$$ and $(b)$ follows \ref{eq:defx1} and \ref{eq:defx2}.

Therefore, the minimum value of $\cvar_{\alpha_2}[X_2]$ is larger than $\cvar_{\alpha_1}[X_1]$, for any given $P_{X_1}$, which completes the first inequality (from left) of the lemma. The last inequality follows by applying the first inequality to $-X_1$ and $-X_2$ in place of $X_1$ and $X_2$. The other inequalities follow from \cite[Lemma 1]{barron2024provable}, which completes the proof.

\subsection{Proof of Theorem \ref{thm:noisy_vcp}} 
\label{proof:thm:noisy_vcp}
Let $P_X(x), P_{X_{noisy}}(x), $ and $P_{X_{vcp}^L}(x)$, denote the distribution of  $X, X_{noisy}$,  and $X_{vcp}^L$, respectively.
    We first prove the claim for the measurement of operators in the computational basis, and then extend it to the measurement wrt to $H$. 
     Let $z \in \{0,2^n-1\}$, where $n$ denotes the number of qubits in the system. Note the following relations:
    \begin{align}\label{eq:compBasis_thm1}
    \bra{z}\tilde{\rho}\ket{z} &= p_0 \bra{z}\rho\ket{z} + \sum_{i=0}^{4^N-1} p_i \bra{z}E_i\rho E_i^\dagger\ket{z}
         \geq p_0\bra{z}\rho\ket{z}\nonumber \\ 
        \bra{z}\tilde{\rho}^{(L)}\ket{z} &= \frac{1}{\sum_{i=0}^{4^N-1}p_i^L} \left(p_0^L \bra{z}\rho\ket{z} + \sum_{i=0}^{4^N-1} p_i^L \bra{z}E_i\rho E_i^\dagger\ket{z}\right)
         \geq \frac{p_0^L \bra{z}\rho\ket{z}}{\sum_{i=0}^{4^N-1}p_i^L} 
    \end{align}
    
    \noindent Since $H$ is a diagonal Hamiltonian, we know there exists a mapping $h: \{0,1\}^n \rightarrow \RR$, such $P_X(x) = \sum_{z: h(z) = x}\bra{z}\rho\ket{z}, P_{X_{noisy}}(x) = \sum_{z: h(z) = x}\bra{z}\tilde{\rho}\ket{z}, $ and $P_{X_{vcp}^L}(x) = \sum_{z: h(z) = x}\bra{z}\tilde{\rho}^{(L)}\ket{z}, $

    \noindent Therefore, the inequality \eqref{eq:compBasis_thm1} implies 
    \begin{align}
        P_{X_{noisy}}(x) & = \sum_{z: h(z) = x}\bra{z}\tilde\rho\ket{z} \geq p_0\sum_{z: h(z) = x}\bra{z}\rho\ket{z} = p_0 P_X(x). \nonumber \\
        P_{X_{vcp}^L}(x) & = \sum_{z: h(z) = x}\bra{z}\tilde\rho^{(L)}\ket{z} \geq \frac{p_0^L }{\sum_{i=0}^{4^N-1}p_i^L}\sum_{z: h(z) = x}\bra{z}\rho\ket{z} = \alpha_L P_X(x).
    \end{align}
    Furthermore,
    \begin{align}
       \bra{z}\tilde{\rho}^{(L)}\ket{z} & = \frac{1}{\sum_{i=0}^{4^N-1}p_i^L} \left(p_0^L \bra{z}\rho\ket{z} + \sum_{i=0}^{4^N-1} p_i^L\bra{z}E_i\rho E_i^\dagger\ket{z}\right) \nonumber \\
       & = \frac{p_0^{L-1} }{\sum_{i=0}^{4^N-1}p_i^L} \left(p_0 \bra{z}\rho\ket{z} + \sum_{i=0}^{4^N-1} p_i \left(\frac{p_i}{p_0}\right)^{L-1}\bra{z}E_i\rho E_i^\dagger\ket{z}\right) \nonumber \\
       & \overset{(a)}{\leq} \frac{p_0^{L} }{p_0\sum_{i=0}^{4^N-1}p_i^L} \left(p_0 \bra{z}\rho\ket{z} + \sum_{i=0}^{4^N-1} p_i \bra{z}E_i\rho E_i^\dagger\ket{z}\right) = \frac{\alpha_L}{p_0} \bra{z}\tilde{\rho}\ket{z}
    \end{align}
    where $(a)$ uses $p_0 \geq p_i$, for all $i \in [1,4^N-1]$. This implies $P_{X_{vcp}^L}(x) \leq \frac{\alpha_L}{p_0} P_{X_{noisy}}(x),$
    Noting that $P_X(x), P_{X_{noisy}}(x), $ and $P_{X_{vcp}^L}(x)$, satisfy the hypotheses of Lemma \ref{lem:cvarIneq}, we obtain the result.

    
\subsection{Proof of Theorem \ref{thm:vcp_comp}}\label{proof:thm:vcp_comp}
Let $P_X(x), \tildeP_{X_L}(x), $ and $\tildeP_{X_M}(x)$, denote the distribution of  $X, \tilde{X}_L$,  and $\tilde{X}_M$, respectively.
     Let $z \in \{0,2^n-1\}$, where $n$ denotes the number of qubits in the system. We know from \eqref{eq:compBasis_thm1}
    \begin{align}\label{eq:compBasis1}
        \bra{z}\tilde{\rho}^{(L)}\ket{z} 
         \geq \frac{p_0^L \bra{z}\rho\ket{z}}{\sum_{i=0}^{4^N-1}p_i^L} 
    \end{align}
    
    \noindent Since $H$ is a diagonal Hamiltonian, we know there exists a mapping $h: \{0,1\}^n \rightarrow \RR$, such $P_X(x) = \sum_{z: h(z) = x}\bra{z}\rho\ket{z}, \tildeP_{X_L}(x) = \sum_{z: h(z) = x}\bra{z}\tilde{\rho}^{(L)}\ket{z}, $ and $\tildeP_{X_M}(x) = \sum_{z: h(z) = x}\bra{z}\tilde{\rho}^{(M)}\ket{z}, $

    \noindent Therefore, the inequality \eqref{eq:compBasis1} implies 
    \[\tildeP_{X_L}(x) = \sum_{z: h(z) = x}\bra{z}\tilde{\rho}^{(L)}\ket{z} \geq \frac{p_0^L }{\sum_{i=0}^{4^N-1}p_i^L}\sum_{z: h(z) = x}\bra{z}\tilde{\rho}^{(L)}\ket{z} = \alpha_L P_X(x).\] Similarly, 
    $\tildeP_{X_M}(x) \geq  \alpha_M P_X(x).$
    Furthermore,
    \begin{align}
       \bra{z}\tilde{\rho}^{(M)}\ket{z} & = \frac{1}{\sum_{i=0}^{4^N-1}p_i^M} \left(p_0^M \bra{z}\rho\ket{z} + \sum_{i=0}^{4^N-1} p_i^M\bra{z}E_i\rho E_i^\dagger\ket{z}\right) \nonumber \\
       & = \frac{p_0^{M-L} }{\sum_{i=0}^{4^N-1}p_i^M} \left(p_0^L \bra{z}\rho\ket{z} + \sum_{i=0}^{4^N-1} p_i^L \left(\frac{p_i}{p_0}\right)^{M-L}\bra{z}E_i\rho E_i^\dagger\ket{z}\right) \nonumber \\
       & \overset{(a)}{\leq} \frac{p_0^{M-L} }{\sum_{i=0}^{4^N-1}p_i^M} \left(p_0^L \bra{z}\rho\ket{z} + \sum_{i=0}^{4^N-1} p_i^L \bra{z}E_i\rho E_i^\dagger\ket{z}\right) = \frac{\alpha_M}{\alpha_L} \bra{z}\tilde{\rho}^{(L)}\ket{z}
    \end{align}
    where $(a)$ uses $p_0 \geq p_i$, for all $i \in [1,4^N-1]$. This implies $\tildeP_{X_M}(x) \leq \frac{\alpha_M}{\alpha_L} \tildeP_{X_L}(x),$
    Noting that $P_X(x), \tildeP_{X_L}(x), $ and $\tildeP_{X_M}(x)$, satisfy the hypotheses of Lemma \ref{lem:cvarIneq}, we obtain the result.


\subsection{Proof of Theorem \ref{thm:noisy_vcp_ratio_obs}}\label{proof:thm:noisy_vcp_ratio_obs}

We begin by computing $\langle X\otimes O\otimes I\rangle $.
Assuming the ancilla is not corrupted by noise, the state of the ancilla is $$\ketbra{+} = \frac{1}{2} \round{\ketbra{0} + |0\rangle\langle1| + |1\rangle\langle0| + \ketbra{1}}.$$ On measuring $X$, we have two terms whose expectations are not zero, one corresponding to $|0\rangle\langle1|$, and another corresponding to $|1\rangle\langle0|$. For the first term, we have the following
\begin{align}
    & \sum_{a,b,c,d,i,j} \!\!\!\!p_{ab}p_{cd}q_iq_j \Tr{ \!(O\!\otimes \!I) (E_{c}\!\otimes\! E_{d})(E_{i}\!\otimes \!E_{j})(E_{a}\!\otimes \!E_{b})\round{\!\rho\otimes \frac{I}{2^n}}\! S^\dagger (E^\dagger_{a}\!\otimes\! E_{b}^\dagger)(E^\dagger_{i}\!\otimes\! E_{j}^\dagger)(E^\dagger_{c}\!\otimes \!E_{d}^\dagger)S} \nonumber \\
    & = \sum_{a,b,c,d,i,j} p_{ab}p_{cd}q_iq_j \Tr{ (O \otimes I) (E_{c}E_{i}E_{a}\otimes E_{d}E_{j}E_{b})\round{\rho\otimes \frac{I}{2^n}} S^\dagger (E^\dagger_{a}E^\dagger_{i}E^\dagger_{c}\otimes E_{b}^\dagger E_{j}^\dagger E_{d}^\dagger)S} \nonumber \\
    & = \sum_{a,b,c,d,i,j} p_{ab}p_{cd}q_iq_j \Tr{ (O \otimes I) (E_{c}E_{i}E_{a}\otimes E_{d}E_{j}E_{b})\round{\rho\otimes \frac{I}{2^n}} (E_{b}^\dagger E_{j}^\dagger E_{d}^\dagger \otimes E^\dagger_{a}E^\dagger_{i}E^\dagger_{c})} \nonumber \\
    & = \sum_{a,b,c,d,i,j} p_{ab}p_{cd}q_iq_j \Tr{ (O \otimes I) (E_{c\oplus i\oplus a}\otimes E_{d\oplus j \oplus b})\round{\rho\otimes \frac{I}{2^n}} (E_{d\oplus j \oplus b}^\dagger \otimes E^\dagger_{c\oplus i\oplus a})} \nonumber \\
    & = \frac{1}{2^n}\sum_{a,b,c,d,i,j} p_{ab}p_{cd}q_iq_j \Tr{O E_{c\oplus i\oplus a}\rho E_{d\oplus j \oplus b}^\dagger  } \Tr{  E_{d\oplus j \oplus b}E^\dagger_{c\oplus i\oplus a}} \nonumber \\
    & = \sum_{a,b,c,d,i,j} p_{ab}p_{cd}q_iq_j \Tr{O E_{c\oplus i\oplus a}\rho E_{d\oplus j \oplus b}^\dagger  } \delta_{d\oplus j \oplus b,c\oplus i\oplus a} \nonumber \\
    & = \sum_{a,b,c,d,k} p_{ab}p_{cd}q_{a\oplus c \oplus k}q_{b\oplus d \oplus k}\Tr{O E_{k}\rho E_{k}^\dagger } = \sum_k \mu_k'\Tr{O E_{k}\rho E_{k}^\dagger },
\end{align}
where
\[
\mu_k' := {\sum_{a,b,c,d} p_{ab}p_{cd}q_{a\oplus c \oplus k}q_{b\oplus d \oplus k}}
\]
The other term corresponding to $|1\rangle\langle0|$ is also given by the same.
\begin{align}
    \sum_k \mu_k'\Tr{O E_{k}\rho E_{k}^\dagger }.
\end{align}
Therefore, 
\begin{align}
    \langle X\otimes O\otimes I\rangle = \sum_k \mu_k'\Tr{O E_{k}\rho E_{k}^\dagger }.
\end{align}
This implies, 
\begin{align}
    \langle X\otimes I\otimes I\rangle = \sum_k \mu_k'\Tr{E_{k}\rho E_{k}^\dagger } = \sum_k \mu_k',
\end{align}
which completes the proof.
    
\subsection{Proof of Theorem \ref{thm:vcp_noisy}}\label{proof:thm:vcp_noisy}
Let $P_X(x), P_{X_{noisy}}(x), $ and $P_{\bar{X}_{vcp}}(x)$, denote the distribution of  $X, X_{noisy}$,  and $\bar{X}_{vcp}$, respectively.
     For $z \in \{0,2^n-1\}$, we have the the following relations:
    \begin{align}\label{eq:compBasis_thm1_mu}
    \bra{z}\tilde{\rho}\ket{z} &= q_0 \bra{z}\rho\ket{z} + \sum_{i=0}^{4^N-1} q_i \bra{z}E_i\rho E_i^\dagger\ket{z}
         \geq q_0\bra{z}\rho\ket{z}\nonumber \\ 
        \bra{z}\calE_\mu(\rho)\ket{z} &= \mu_0 \bra{z}\rho\ket{z} + \sum_{i=0}^{4^N-1} \mu_i \bra{z}E_i\rho E_i^\dagger\ket{z} \geq \mu_0\bra{z}\rho\ket{z}
    \end{align}
    

    We know 
    \begin{align}
        P_{X_{noisy}}(x) & = \sum_{z: h(z) = x}\bra{z}\tilde\rho\ket{z} \geq q_0\sum_{z: h(z) = x}\bra{z}\rho\ket{z} = q_0 P_X(x). \nonumber \\
        P_{\bar{X}_{vcp}}(x) & = \sum_{z: h(z) = x}\bra{z}\calE_\mu(\rho)\ket{z} \geq 
        \mu_0 P_X(x).
    \end{align}
    Furthermore,
    \begin{align}
       \bra{z}\calE_\mu(\rho)\ket{z} & = \mu_0^L \bra{z}\rho\ket{z} + \sum_{i=0}^{4^N-1} \mu_i\bra{z}E_i\rho E_i^\dagger\ket{z} \nonumber \\
       & = \frac{\mu_0}{q_0} \left(q_0 \bra{z}\rho\ket{z} + \sum_{i=0}^{4^N-1} q_i \left(\frac{\mu_i}{\mu_0}\frac{q_0}{q_i}\right)\bra{z}E_i\rho E_i^\dagger\ket{z}\right) \nonumber \\
       & \overset{(a)}{\leq} \frac{\mu_0}{q_0} \left(q_0 \bra{z}\rho\ket{z} + \sum_{i=0}^{4^N-1} q_i \bra{z}E_i\rho E_i^\dagger\ket{z}\right) = \frac{\mu_0}{q_0} \bra{z}\tilde\rho\ket{z}
    \end{align}
    where $(a)$ uses $\frac{\mu_i}{\mu_0} \leq \frac{q_i}{q_0}$, for all $i \in [1,4^N-1]$. This implies $P_{\bar{X}_{vcp}}(x) \leq \frac{\mu_0}{q_0} P_{X_{noisy}}(x),$
    Since $P_X(x), P_{X_{noisy}}(x), $ and $P_{\bar{X}_{vcp}}(x)$, satisfy the conditions of Lemma \ref{lem:cvarIneq}, we have the result.





\bibliographystyle{IEEEtran}
\bibliography{references}

\end{document}

%% file: header.tex
\newtheorem{mydefinition}{Definition}
\newtheorem{mylemma}[mydefinition]{Lemma}

\newtheorem{mytheorem}[mydefinition]{Theorem}
\newtheorem{myremark}[mydefinition]{Remark}

\global\long\def\RR{\mathbb{R}}

\global\long\def\EE{\mathbb{E}}
\global\long\def\PP{\mathbb{P}}

\def \calD {\mathcal{D}}
\def \calE {\mathcal{E}}

\def \calH {\mathcal{H}}
\def \calI {\mathcal{I}}

\def \calP {\mathcal{P}}

\def \calU {\mathcal{U}}

\def \calX {\mathcal{X}}

\def \tildeP {\tilde{P}}
\def \tildeX {\tilde{X}}

\def \RR {\mathbb{R}}

\newcommand{\round}[1]{\left({#1}\right)}

\def \cvar {\textbf{CVaR}}
\def \tcvar {\text{CVaR}}
\def \bcvar {\overline{{\textbf{CVaR}}}}

\usepackage{stackengine}



%% file: example_arxiv.tex
\subsection{Depolarizing noise}\label{sec:example1} 
In this example, we intend to purify a Clifford Unitary $C$ that has the error $\calE_C$ modeled by a depolarizing noise of parameter $q$. The noise in the CSwap circuit $\calE_S$ is also assumed to be depolarizing in nature with an independent parameter $p$. For the ease to illustration, we consider $N=2$, i.e., $C$ acts on two qubits. This implies
\begin{align}
    \calE_S(\rho) & = (1-p)\rho + p\frac{I_{2N}}{2^{2N}} = (1-p)\rho + \frac{p}{4^{2N}}\sum_{i=0}^{4^{2N}-1}E_i \rho E_i^\dagger \nonumber \\
    \calE_C(\rho) 
     & = (1-q)\rho + q\frac{I_N}{2^N} 
\end{align}
The first constraint from the Theorem \ref{thm:vcp_noisy} is  $\mu_0 \geq q_0$. 
We know 
\begin{align}
    \mu_0 &= \frac{\sum_{a,b,c,d \in \{0,1\}^4}p_{ab}p_{cd}q_{a\oplus c}q_{b\oplus d}}{\sum_{a,b,c,d,k \in \{0,1\}^4}p_{ab}p_{cd}q_{a\oplus c\oplus k}q_{b\oplus d \oplus k}}
\end{align}
For $\calE_S$ and $\calE_C$, we have 
\begin{align}
p_{ab} &= 
    \begin{cases}
        \round{1-p\round{1-\frac{1}{256}}} & \text{ if} \; (a,b) = (\vec{0},\vec{0})\\
        \cfrac{p}{256} & \text{ otherwise}
    \end{cases}\nonumber\\
q_{i} &= 
    \begin{cases}
        \round{1-q\round{1-\frac{1}{16}}} & \text{ if} \; i = \vec{0}\\
        \cfrac{q}{16} & \text{ otherwise}
    \end{cases}
\end{align}
This gives
\begin{align}
    \sum_{a,b,c,d \in \{0,1\}^4}p_{ab}p_{cd}q_{a\oplus c}q_{b\oplus d} &= \frac{1}{256} \big( 256 - 510 p + 255 p^2  - 480 (-1 + p)^2 q + 225 (-1 + p)^2 q^2 \big), \nonumber\\
    \sum_{a,b,c,d \in \{0,1\}^4}p_{ab}p_{cd}q_{a\oplus c\oplus k}q_{b\oplus d\oplus k} &= \frac{1}{256} \big( q^2  + (-2 + p) p (-1 + q^2) \big) \text{ for all } k\in [1,2^{4}]. \nonumber
\end{align}
This gives 
\begin{align}
    \mu_0 = \frac{256 - 510 p + 255 p^2 - 480 (1-p)^2 q + 225 (1-p)^2 q^2}{16 \left( 16 - 30 p (1-q)^2 + 15 p^2 (1-q)^2 - 15 (2 - q) q \right)}\nonumber
\end{align}
Note that when $\calE_S$ and $\calE_C$ are depolarizing, the resulting $\calE_\mu$ is also depolarizing. This implies, $\mu_i = \frac{1-\mu_0}{4^N}$ for all $i\neq 0$. Further, if $\mu_0 \geq q_0$, this implies 
\[
\frac{\mu_i}{\mu_0} = \frac{1-\mu_0}{4^N\mu_0} \leq \frac{1-q_0}{4^N q_0} = \frac{q_i}{q_0}.
\]
Therefore, we only need to examine the relation between $\mu_0$ and $q_0$.
We note that 
\[
\mu_0 \geq q_0 \iff q_l(p) \leq q \leq \min\{q_u(p) ,1\} , \; 0 \leq p \leq p_l 
\]
where
\begin{figure}
    \centering
    \vspace{-5pt}
    \includegraphics[width=0.5\linewidth]{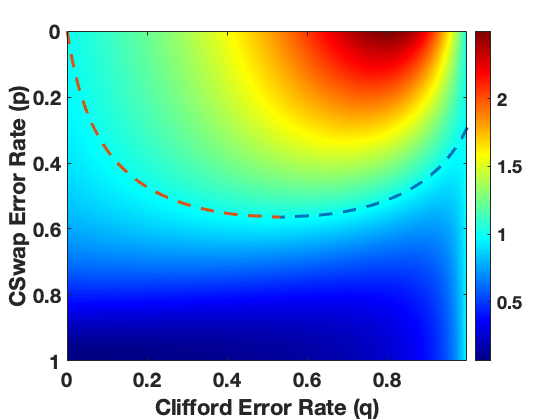}
    \caption{\emph{Ratio of $\mu_0$ to $q_0$
  for different values of the parameters 
$p$ and $q$}. The plot shows the ratio $\mu_0/q_0$ as a function of 
$p$ and $q$. The red and blue dashed curves are the plots for $q_l(p)$ and $q_u(p)$ as a function of $p$. Moreover, on the curve, we have the ratio as $1$ or $\mu_0=q_0$. The curve divides the plot into two regions: the upper region (above the curve) where $q_l(p) \leq q \leq q_u(p)$ and $0\leq p\leq p_u$, and thus the ratio exceeds one, and the lower region (below the line) where the ratio is less than one. }
\vspace{-10pt}
    \label{fig:noise_dep_analytical}
\end{figure}
\begin{align}
 q_l(p) &:=\frac{1}{15} \left( 8 - \frac{\sqrt{ \left( 64 + 79 (p-2) p \right)}}{(1- p)} \right) \nonumber \\
 q_l(p) &:=\frac{1}{15} \left( 8 + \frac{\sqrt{ \left( 64 + 79 (p-2) p \right)}}{(1- p)} \right) \nonumber \\
    & p_l := 1 - \frac{ \sqrt{1185}}{79} = 0.56425.\nonumber
\end{align}
Therefore, from Theorem \ref{thm:vcp_noisy}, the combination of CVaR and VCP will be advantageous in the regime when 
\[q_l(p)  \leq q \leq \min\{q_u(p) ,1\} ,\quad \text{and} \quad 0 \leq p \leq p_l .\] 
Note that when the error in CSwap gates in the VCP protocol is not accounted for, the only constraint on $q$ is $q\leq 1$. However, when the error is accounted, $q$ is constrained by an upper and a lower bound. 

Figure \ref{fig:noise_dep_analytical} shows the ratio of $\mu_0$ and $q_0$ for varying values of $p$ and $q$. The line in the middle of the plot separates the region when the ratio is more than one to when it is less than one. The region above the line corresponds to $\mu_0 \geq q_0$, i.e., the values of $(p,q)$ that allows the scheme to be beneficial from the perspective of mitigating the noise.




\begin{figure}[!htb]
\centering
\vspace{-5pt}
  \includegraphics[width=0.5\linewidth]{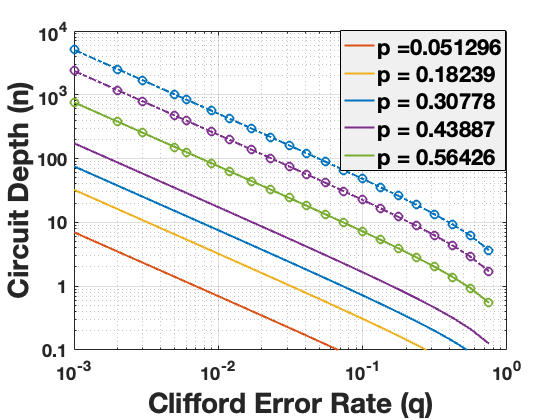}
  \vspace{-5pt}
    \caption{\emph{Lower and upper bounds on $n$ for different values of $p$ with varying $q$}.  The solid lines represent the lower bounds for each value of 
$p$, while the upper bounds are marked with circle markers, with no connecting lines. Each bound (lower and upper) is color-coded according to the corresponding value of the parameter $p$, allowing for easy comparison across different parameter values. The missing upper bound curves for some values of 
$p$ are due to the bound being infinite (trivial) for such $p$, which is reflected by the absence of markers at those locations. The x-axis represents the parameter $q$ and both axes are plotted on a logarithmic scale.}
    \label{fig:num_cnots_loglog}
\end{figure}

\subsection{Collection of IID gates}\label{sec:example2}
As another example, we consider a circuit that is constructed using a collection of $n$ two-qubit gates that are independent and identical in their error profile, with each having a depolarizing error of parameter $q$. The CSwap error is also assumed to be depolarizing with parameter $p$. To analyze this scenario, we note that all the gates and noise components can be grouped together. The gives the effective noise as a depolarizing noise with parameter $q_{eff} = 1-(1-q)^n$. Using the Example in Section \ref{sec:example1}, for a given $p$ and $q$, we can provide bounds on $n$ when a noisy-VCP still performs better than a no VCP technique. 
We need $p\geq p_l$ and $q_l(p)  \leq q_{eff} \leq \min\{q_u(p) ,1\}$, which gives
\begin{align}
    \left\lceil\frac{\log(1-q_l(p) )}{\log(1-q)}\right\rceil \leq n \leq  \left\lfloor\frac{\log(1-\min\{q_u(p),1\} )}{\log(1-q)}\right\rfloor
\end{align}
Figure \ref{fig:num_cnots_loglog} shows the variation of the above lower and upper bounds for different values of $p$ and $q$.